\documentclass[a4paper,11pt]{article}
\usepackage{amsmath, amssymb, amsthm}
\usepackage{hyperref}
\usepackage[margin=1in]{geometry}
\newtheorem{definition}{Definition}
\newtheorem{theorem}{Theorem}
\newtheorem{lemma}{Lemma}
\newtheorem{corollary}{Corollary}

\title{Improved lower bounds on the time complexity of Shellsort}
\author{Zhenghan Zang}
\date{}

\begin{document}

\maketitle

\begin{abstract}
In this paper we develop the framework based on the parametrized mapping
\begin{align*}
    [\sigma(1), \sigma(2), \cdots, \sigma(n)] \mapsto \sigma(1)z + \sigma(2)z^2 + \cdots \sigma(n)z^n
\end{align*} to perform runtime analysis on Shellsort. In particular, we show that the worst-case time complexity of Shellsort using Tokuda's sequence \cite{Tokuda1992} is at least $\Omega(N^{1.26})$ with a generalisation of this result to any strictly decreasing gap sequence where each term at most a fixed distance away from a geometric sequence with a rational common ratio, and we also show that strictly decreasing gap sequences giving worst-case Shellsort time complexities of $O(N \log^c N)$ must have $\Omega(\log N / \log \log N)$ terms of order $\Omega(N / (\log N)^c)$. 
\end{abstract}

\section{Introduction}
\noindent Since its introduction by Shell in \cite{Shell1959}, the time complexity of Shellsort has been a long standing problem in computer science. For the past few decades, many of the most important works on the runtime analysis of Shellsort have been done using combinatorial and number theoretic methods. In \cite{Pratt1972}, Pratt showed that Shellsort using any gap sequence which approximates a geometric progression with integer coefficients has worst case time $\Omega(N^{3/2})$ as well as constructing a famous gap sequence consisting of numbers of the form $2^p3^q$ which gives both worst case and average time of $\Theta(N (\log N)^2)$, yet this sequence is seldomly used in practical implementations due to its inefficiency to compute which comes from its high gap density. \newline

\noindent For a long time it was unknown whether there existed gap sequences of length $O(\log N)$ whose worst case time complexities would break the $\Omega(N^{3/2})$ barrier, Sedgewick in \cite{Sedgewick1986} (having announced it in 1982) was the first to construct such a sequence which consists of numbers of the form $4^{j+1}+3\cdot2^j+1$, in the same paper he proved that Shellsort using this gap sequence has a worst-case time complexity of $O(N^{4/3})$. Based on this work, in \cite{SI1985} Incerpi and Sedgewick showed that for any $\epsilon > 0$ there exists a gap sequence of length $O_\epsilon(\log N)$ such that Shellsort using it has a worst case time complexity of $O(N^{1+\epsilon/\sqrt{\log N}})$. In a sense this is the optimal, as in \cite{PPS1993} it is proven that the Shellsort on array of length $N$ using any static gap sequence of length $m$ has worst case time at least $N^{1+c/\sqrt{m}}$ for some fixed $c>0$, provided that $N$ is large enough compared to $m$ (but $m = O(\log N)$ is sufficient). In the same paper it is also proven that Shellsort using any static gap sequence has time complexity of at least $\Omega(N(\log N /\log \log N)^2)$, a bound which has been later improved to $\Omega(N(\log N )^2/\log \log N)$ but for Shellsort sorting networks only in \cite{Cypher1993}. \newline

\noindent Although the Incerpi-Sedgewick gap sequences proposed in \cite{SI1985} give almost asymptotically optimal worst-case times, these aren't widely used in practice due to their practical inefficiencies. For the last few decades, research has been conducted on Shellsort to empirically look for gap sequences that perform well in practical scenarios, with sequences being proposed such as Tokuda's sequence introduced in \cite{Tokuda1992}, which consists of numbers of the form $\lceil (9\cdot (9/4)^k-4)/5\rceil$, and Ciura's sequence introduced in \cite{Ciura2001}. \newline

\noindent However, no version of Shellsort using one of these emperically derived gap sequences has any known individually proven nontrivial lower/upper bound stronger than the ones implied by \cite{PPS1993}. For many years, studying the Frobenius coin problem or some variant of it have been the standard approach when it comes to runtime analysis of Shellsort. Research based on this seems to have reached a bottleneck in recent years due to the computational intractability of the solutions to the Frobenius problem when given large sets of arguments and also that it isn't a good enough representation of Shellsort. This especially becomes a problem when the gap sequences to be dealt with are too irregular. \newline

\noindent In this paper we prove that the worst case time complexity of sorting an array of length $N$ using Tokuda's sequence is at least $\Omega(N^{1.26})$, by proving a more general form of this result for any gap sequence of at most a fixed distance away from a rational geometric sequence (see theorem 3, section 2). Additionally, we prove that strictly decreasing gap sequences giving worst-case Shellsort time complexities of $O(N \log^c N)$ must have $\Omega(\log N / \log \log N)$ terms of order $\Omega(N / (\log N)^c)$ (see theorem 6, section 2). We do this by introducing a new mechanism to perform runtime analysis on Shellsort which uses a parameterized mapping that takes a permutation $\sigma$ of the numbers $1, 2, \cdots, N$ and returns a polynomial about a complex variable $z$ (see definition 2, section 1.1). Our focus is on the effects the swaps performed when running Shellsort have on the values of this potential function for carefully chosen (according to the gap sequence given) values of $z$, by viewing each swap as an operation taking one permutation to another. This approach bypasses many of the issues faced by Frobenius-based approaches. We begin by making the following definitions. 

\subsection{Notations}
In this paper, denote by $N$ any positive integer greater than $1$ and we assume that we perform Shellsort on an array of size $N$ which contains each of $1, 2, \cdots, N$ exactly once. When given $N$, denote $L$ a (strictly) decreasing gap sequence with the first term less than $N$ and $\sigma$ a permutation of $1, 2, \cdots, N$, and let $id$ be the identity permutation where $id(i) = i$ for all $1 \leq i \leq N$. Also for a real number $r$, denote by $||r||$ the distance from $r$ to the nearest integer. 
\begin{definition}
We define $T(N, \sigma, L)$ as the number of swaps needed to sort the array $[\sigma(1), \sigma(2), \cdots, \sigma(N)]$ using shellsort with gap sequence $L$. 
\end{definition}

\begin{definition}
For any complex number $z$, we define the potential of the permutation $\sigma$ with respect to $z$ as the polynomial 
    \begin{align*}
        p(N, \sigma, z) = \sum_{i=1}^N \sigma(i)z^i 
    \end{align*}
\end{definition}

\begin{definition}
We call an operation defined on any permutation $\sigma$ a swap if it returns another permutation $\sigma'$ that is the same as $\sigma$ except having $\sigma'(i) = \sigma(j), \sigma'(j) = \sigma(i)$ for some $1 \leq i < j \leq N$. Denote this operator as $S_{i,j}$ and denote the distance of a swap as $dist(S_{i,j}) = j - i$. 
\end{definition}

\begin{definition}
For any $1 \leq x \leq N$, denote $l(x,N)$ as the number of terms $s$ in $L$ satisfying $x \leq s < N$.
\end{definition}
    
\section{Main Results}

\noindent We begin with the following theorem which constitutes the core idea of this paper. 

\begin{theorem}
For any $N$, a strictly decreasing gap sequence $L$ and any real numbers $r, x_0 > 0, \alpha \geq 0$, which satisfies that for any $s$ in $L$ with $s \geq x_0$ we have $||sr|| \leq \alpha$ (Here we use $||x||$ to denote the distance from $x$ to its nearest integer), then for any permutation $\sigma$ of $1, 2, \cdots, N$ we have
\begin{align}
    T(N, \sigma, L) \geq \frac{|p(N, \sigma, z) - p(N, id, z)|}{2(x_0 + \pi \alpha N)}
\end{align}
where $z = \exp (2\pi ir)$. 
\end{theorem}

\begin{proof}
\noindent In fact we will prove a stronger result: let $S_{i_1, j_1}, S_{i_2, j_2}, \cdots, S_{i_t,j_t}$ be a sequence of swaps of minimum length such that $(S_{i_t, j_t} \circ \cdots \circ S_{i_2, j_2} \circ S_{i_1, j_1})(\sigma) = id$, where the distance of the swaps $j_k - i_k$ is non-increasing and all these distances belong to the gap sequence $L$, then denoting this $t$ by $t(N, \sigma, L)$ (we'll from now on write $t$ for short) the same bound in (1) holds with $T(N, \sigma, L)$ replaced by $t$. Note that since the sequence of swaps being carried out when running shellsort can be written as such a sequence (though of not necessarily the minimum length) we must have $T(N, \sigma, L) \geq t$, so proving this result would directly imply the theorem above. \newline 

\noindent Let $u$ be the largest number such that $j_u - i_u \geq x_0$ (and $u=0$ if such a number doesn't exist) and let $\omega = (S_{i_u, j_u} \circ \cdots \circ S_{i_2, j_2} \circ S_{i_1, j_1})(\sigma)$. We begin by noting that for any complex number $z$ we have $|p(N, \sigma, z) - p(N, id, z)| \leq |p(N, \omega, z) - p(N, id, z)| + |p(N, \sigma, z) - p(N, \omega, z)|$. \newline

\noindent Bounding the first difference is easy: just note that whenever $|z| = 1$ (which is the case for the $z$ we have defined):  
    \begin{align*}
        \begin{split}
        |p(N, \omega, z) - p(N, id, z)| & = |\sum_{i=1}^N \omega(i)z^i - \sum_{i=1}^N iz^i| \\ & = |\sum_{i=1}^N (\omega(i) - i)z^i|
         \\ & \leq \sum_{i=1}^N|\omega(i) - i| \\ & \leq 2tx_0
        \end{split}
    \end{align*}
the last inequality holds since the sum of the changes in positions of each $i$ going from $[1, 2, \cdots, N]$ to $[\omega(1), \omega(2), \cdots, \omega(N)]$ is at most $2(t-u)x_0 \leq 2tx_0$. \newline

\noindent Now note that for any $|z| = 1$ we have $|z - 1| = 2|\sin (\pi ||\arg (z)/(2\pi)||)|$. Now since the distance from any real number to its nearest integer is between $0$ and $1/2$ it can be verified that $|z - 1|/||\arg (z)/(2\pi)|| =2|\sin (\pi ||\arg (z)/(2\pi)||)|/||\arg (z)/(2\pi)||$ lies bet en $4$ and $2\pi$. This therefore gives $|z -1| \leq 2\pi||\arg (z)/(2\pi)||$ for all unit complex numbers $z$. Now setting $z = \exp (2\pi i r)$, since for any $s$ in $L$ with $s \geq x_0$ we have $|z^s - 1| \leq 2\pi ||sr|| \leq 2\pi \alpha$.  \newline

\noindent Using this $z$ we can now bound $|p(N, \sigma, z) - p(N, \omega, z)|$. Note that for any permutation $\sigma'$ and a swap $S_{k,k+s}$ of distance $s$ where $s \in L, x_0 \leq s < N$ we have 
    \begin{align*}
        \begin{split}
        |p(N, S_{k,k+s}(\sigma'), z) - p(N, \sigma', z)| & = |\sum_{i=1}^N (S_{k,k+s}(\sigma')(i) - \sigma'(i))z^i|
        \\ &= |\sigma'(k)z^{k+s}+\sigma'(k+s)z^k - \sigma'(k)z^k - \sigma'(k+s)z^{k+s}| 
        \\ &= |\sigma'(k) - \sigma'(k+s)||z^s - 1|
        \\ &< 2\pi \alpha N
        \end{split}
    \end{align*}

\noindent Therefore as $\omega = (S_{i_u, j_u} \circ \cdots \circ S_{i_2, j_2} \circ S_{i_1, j_1})(\sigma)$, since for every $1 \leq v \leq u$ we have $j_v - i_v \in L, x_0 \leq j_v-i_v \leq N$, this means that
    \begin{align*}
        \begin{split}
            |p(N, \sigma, z) - p(N, \omega, z)| &= |\sum_{v=1}^u (p(N, (S_{i_v, j_v}\circ \cdots \circ S_{i_1, j_1}) (\sigma), z) - p(N, (S_{i_{v-1}, j_{v-1}}\circ \cdots \circ S_{i_1, j_1}) (\sigma), z))| \\
            &\leq \sum_{v=1}^u |p(N, (S_{i_v, j_v}\circ \cdots \circ S_{i_1, j_1}) (\sigma), z) - p(N, (S_{i_{v-1}, j_{v-1}}\circ \cdots \circ S_{i_1, j_1}) (\sigma), z)|
            \\ &\leq 2u\pi \alpha N
            \\ & \leq 2t \pi \alpha N
        \end{split}
    \end{align*}

\noindent Combining the bounds we've obtained so far gives
    \begin{align*}
        \begin{split}
            |p(N, \sigma, z) - p(N, id, z)| &\leq |p(N, \omega, z) - p(N, id, z)| + |p(N, \sigma, z) - p(N, \omega, z)| 
            \\ &\leq 2tx_0 + 2t \pi \alpha N
            \\ &= 2t(x_0 + \pi \alpha N)
        \end{split}
    \end{align*}
\newline

\noindent So for any permutation $\sigma$ of $1, 2, \cdots, N$, 
    \begin{align*}
        t(N, \sigma, L) \geq \frac{|p(N, \sigma, z) - p(N, id, z)|}{2(x_0 + \pi \alpha N)}
    \end{align*}
\noindent and the theorem follows. (Note: equality of (1) may be achieved iff $\sigma = id$.)
\end{proof}

\noindent Note that the above result can be viewed as in a sense a generalisation of theorem 2.11 in Pratt's paper \cite{Pratt1972}. To see this we'll first need the following lemma. 

\begin{lemma}
There exists an absolute constant $c_0 > 0$ such that for any $n$ unit complex numbers $z_1, z_2, \cdots, z_n$, there exists a permutation $\omega$ of $1, 2, \cdots, n$ such that $|\omega(1)z_1+\omega(2)z_2+\cdots +\omega(n)z_n| \geq c_0n^2$. 
\end{lemma}

\noindent The proof of this is quite intuitive - the idea being that by assigning the larger values to $z_i$s that are somewhat concentrated we can force the resulting sum to be of order $\Omega(n^2)$. The full proof of this is simple yet quite tedious and is therefore left in the appendix. \newline

\noindent Note what this implies is that there is an absolute constant $c_1 > 0$ such that as long as $||\arg (z)/(2\pi)|| \geq c_1/N$ then there is a permutation $\sigma$ of $1, \cdots, N$ such that $|p(N, \sigma, z) - p(N, id, z)| \geq C_0N^2$ for some universal constant $C_0$ that is independent of $N$. To see this, if $|p(N, id, z)| = |z(1-z^N)/(1-z)^2-Nz^{N+1}/(1-z)^2| > c_0N^2/2$ (where $c_0$ same as defined in lemma 1) then $2/|1-z|^2+N/|1-z| > c_0N^2/2$, so this gives us that $|1-z| < c_1'/N$ for some absolute positive constant $c_1'$ and so $||\arg(z)/(2\pi)|| < c_1/N$ for some absolute constant $c_1 > 0$. And if $|p(N, id, z)| \leq c_0N^2/2$ then letting $\sigma = \omega$ where $\omega$ is the same as defined in lemma 1 gives $|p(N, \sigma, z) - p(N, id, z)| \geq c_0N^2/4$. So letting $C_0 = c_0/4$ works. \newline

\noindent Combining this with theorem 1 and lemma 1 gives:
\begin{corollary}
    There are universal constants $C_0, c_1 > 0$ such that for any $N > 1$, a strictly decreasing gap sequence $L$ and any real numbers $r, ||r|| \geq c_1/N, x_0 > 0, \alpha \geq 0$, which satisfies that for any $s$ in $L$ with $l \geq x_0$ we have $||sr|| \leq \alpha$, then there exists a permutation $\sigma$ of $1, 2, \cdots, N$ where
\begin{align}
    T(N, \sigma, L) \geq \frac{C_0N^2}{2(x_0 + \pi \alpha N)}
\end{align} 
\end{corollary}

\noindent From this corollary we obtain that:
\begin{theorem}
    For any reals $0 < \lambda_1 < \lambda_2, \mu > 0,0 < \nu, \eta < 1$ there exists a positive constant $C(\lambda_1, \lambda_2, \mu, d) > 0$ depending only on $\lambda_1, \lambda_2, \mu$ and $d$ such that for any large enough $N$ and a strictly decreasing gap sequence $L$, if there exists a real number $\lambda_1N^\nu \leq r' \leq \lambda_2N^\nu$ and a non-negative $d \geq 0$ such that for any $s \geq \mu N^{1-\eta}$ in $L$ there exists an integer $s'$ such that $|s - s'r'| \leq d$, then there is a permutation $\sigma$ of $1, 2, \cdots, N$ such that 
    \begin{align}
        T(N, \sigma, L) > \frac{C(\lambda_1, \lambda_2, \mu, d)}{d}N^{1+\min (\nu, \eta)}
    \end{align}
    In particular, when $\nu = \eta = 1/2$ we have 
    \begin{align}
        T(N, \sigma, L) > \frac{C(\lambda_1, \lambda_2, \mu, d)}{d}N^{3/2}
    \end{align}
\end{theorem}

\begin{proof}
    This is trivial by corollary 1: just set $x_0 = \mu N^{1 - \eta}, r = 1/r', \alpha = d/r'$ in (2) and that's it. 
\end{proof}

\noindent Note that (4) is in a certain sense stronger than theorem 2.11 of \cite{Pratt1972} in that the conditions don't actually require the existence of a term in $L$ of order $\sqrt{N}$, nor does it assert any preconditions on the terms in $L$ less than $r$. On the other hand for $T(N, \sigma, L)$ to be of order of at least $N^{3/2}$ we require $d$ to be bounded by some constants, while in theorem 2.11 in \cite{Pratt1972} this is more lenient. In particular, (4) shows that any gap sequence whose terms are of at most a fixed constant away from some geometric progression with integer common ratio must be a worst case time of order at least $N^{3/2}$. \newline 

\noindent With corollary 1 we may in fact also deal with gap sequences which approximate a geometric progression with a noninteger rational common ratio such as Tokuda's gap sequence \cite{Tokuda1992}. In particular:

\begin{theorem}
    Suppose that $a>b>1$ are coprime integers and $q,d > 0$ be any fixed positive real numbers. For each $N$ let $L_N$ be a decreasing gap sequence such that for any $s$ in $L_N$, there exists an integer $n$ such that $|s - q(a/b)^n| < d$, then
    \begin{align}
        \max_{\sigma} T(N,\sigma, L_N) = \Omega(N^{1+\frac{\log_ba-1}{2\log_ba-1}})
    \end{align}
\end{theorem}

\begin{proof}
     Let $n_0$ be the largest integer such that $q(a/b)^{n_0} - d < N$. Consider $r = b^{n_0}/(qa^m)$ where $m$ is an integer to be determined. Now if $q(a/b)^n - d < s < q(a/b)^n + d$, then $a^{n-m}b^{n_0-n} - dr < sr < a^{n-m}b^{n_0-n} + dr$. So whenever $n_0 \geq n \geq m$ (meaning that $s \geq q(a/b)^m + d$) we'd have $||sr|| < dr < db^{n_0}/(qa^m) = \alpha$ and let $x_0 = q(a/b)^m+d$. Now just note that $x_0 + 2\pi\alpha N < 2\pi(x_0+\alpha N) = O_{q,d}((a/b)^m+N(b^{n_0}/a^m))$. To make this small consider when $a^{m_0}/b^{m_0} = Nb^{n_0}/a^{m_0}$. Letting $a = b^{\gamma}$, then this rearranges to $b^{m_0(2\gamma - 1) - n_0} = N$, or $m_0 = (\log _bN + n_0)/(2\gamma -1)$. Setting $m = \lfloor m_0 \rfloor$ gives $(a/b)^m+N(b^{n_0}/a^m) = O_{a,b,q}(N^{\gamma/(2\gamma -1)})$. Now note that $r = \Theta(N^{(1-\gamma)/(2\gamma-1)})$ so for any large $N$ we must have $||r|| > c_1/N$. So by corollary 1 we have for the worst case time, $\max_{\sigma} T(N,\sigma, L_N) = \Omega(N^{2 - \gamma/(2\gamma -1)}) = \Omega(N^{(3\gamma - 2)/(2\gamma-1)})$, and we're done. 
\end{proof}

\noindent If we apply this to Tokuda's sequence, we see that it has a worst case time of at least $\Omega(N^{1.2695...})$. Also note an interesting phenomenon: if we set $b = 1$, then $\log_ba = +\infty$ and that gives $\max T = \Omega (N^{3/2})$ which matches what we got in theorem 2.  \newline

\noindent A recent preprint by Liu \cite{Liu2026} proposed a new gap sequence empirically discovered using machine learning which in the tests they carried out had the best average operation count among all classical gap sequences. In their paper they used the bound in theorem 3 above (which had been established in an earlier version of this preprint) for the worst-case time complexity, and they also showed that for their sequence the polynomial exponent of this bound is tight; it is off by at most a factor of $O(\text{polylog }N)$. It is worth investigating when bounds given by corollary 1 above are tight. \newline

\noindent Also note that if treated more carefully the proof of theorem 3 above would give a more explicit bound where the constant implied by the $\Omega$ symbol can be expressed using $q, d, a, b$. Theorem 3 in this way can potentially be generalized to gap sequences whose terms are at most some constant away from any fixed geometric sequence including those whose common ratios are irrational, provided this ratio is a positive real number greater than $1$ - by approximating this using a sequence of rationals and using irrationality measures.  \newline

\noindent Now, we want to bound the sizes of $x_0, \alpha$ in theorem 1 such that a corresponding $r$ that's not too close to an integer is guaranteed to exist.

\begin{theorem}
There exist some constants $c > 0$ such that, for any given $N > 4$, strictly decreasing gap sequence $L$, and for any $1 \leq x < N$ where $l(x, N) \geq 1$, there exists a complex number $z$ satisfying $z^N = 1$ and $|1 - z| \geq cN^{-1/2}$ for some universal positive constant $c$ such that 
    \begin{align}
        T(N, \sigma, L) \geq \frac{1}{4\pi} \cdot \frac{|p(N, \sigma, z) - p(N, id, z)|}{x + N^{1-\frac{1}{2l(x,N)}}}
    \end{align}
for any permutation $\sigma$ of $1, 2, \cdots, N$. 
\end{theorem}

\begin{proof}
We first need a stronger version of Dirichlet's theorem on Diophantine approximation.

\begin{lemma}
Let $a_1, \cdots, a_m$ be integers satisfying $1 \leq a_i \leq N$. For any positive integers $Q$, $M$, if $2MQ^m \leq N$, then there exists some real number $r$ satisfying $||r|| \geq M/N$ such that $Nr$ is an integer, $||a_ir||\leq 1/Q$ for all $1 \leq i \leq m$. 
\end{lemma}

\begin{proof}
We work in $m$ dimensions. Consider the points
    \begin{align*}
        v_k = \Big(\Big\{\frac{ka_1}{N}\Big\}, \cdots, \Big\{\frac{ka_m}{N}\Big\} \Big)
    \end{align*}
where $k = 0, 1, \cdots, MQ^m$. Here $\{\cdot \}$ represents the fractional part. \newline

\noindent Split $[0, 1)^m$ into $Q^m$ identical small $m$-dimensional hypercubes, each of side length $1/Q$. Since there are $MQ^m+1$ such points $v_k$, by the pigeonhole principle, there must be some small hypercube of side length $1/Q$ which contains at least $M+1$ of the $v_k$'s. We may therefore pick numbers $k, k'$, where $k > k'$, $v_k, v_{k'}$ lie in the same small hypercube and $q = k - k' \geq M$. This means that for every $i$ we'd have 
    \begin{align*}
        \Big|\Big|\frac{qa_i}{N}\Big|\Big| \leq \Big|\Big\{\frac{ka_i}{N}\Big\} - \Big\{\frac{k'a_i}{N}\Big\} \Big| \leq \frac{1}{Q}
    \end{align*}
\newline

\noindent Let $r = q/N$. We then have that $Nr = q \in \mathbb{Z}, ||a_ir|| = ||qa_i/N|| \leq 1/Q$ for all $1 \leq i \leq m$. Now as $q \geq M$ we have $r = q/N \geq M/N$. Since $q \leq MQ^m \leq N/2$, so $r \in [0, 1/2]$, and therefore $||r|| = r \geq M/N$.
\end{proof}

\noindent Now fix a random $1 \leq x < N$ where $l(x, N) \geq 1$. In the above lemma, we let $m = l(x, N), Q=\lfloor N^{\frac{1}{2l(x,N)}} \rfloor > \frac{N^{\frac{1}{2l(x,N)}}}{2},  M = \lfloor \frac{N}{2Q^m} \rfloor \geq \lfloor \sqrt{N}/2 \rfloor > \frac{\sqrt{N}}{4}$ since $N > 4$. This therefore gives that there is a real number $r$ with $Nr \in \mathbb{Z}$ such that for any $s \in L$ where $x \leq s < N$ we have $||sr|| < 2N^{-\frac{1}{2l(x,N)}}$, while $||r|| > N^{-1/2}/4$. Now let $z = \exp({2\pi ir})$. This gives $z^N = 1, |z - 1| > cN^{-1/2}$ for some fixed universal constant $c >0$. \newline

\noindent So letting $x_0 = x, \alpha = 2N^{-\frac{1}{2l(x,N)}} $, by theorem 1 we have for our chosen $z$, that $|z| = 1, |z - 1| > cN^{-1/2}$, and for any permutation $\sigma$ of $1, 2, \cdots, N$, 
    \begin{align*}
        t(N, \sigma, L) \geq \frac{1}{4\pi} \cdot \frac{|p(N, \sigma, z) - p(N, id, z)|}{x + N^{1-\frac{1}{2l(x,N)}}}
    \end{align*}
\noindent and the theorem follows.  
\end{proof}

\noindent We may reformulate the theorem in the following manner:
\begin{corollary}
There exists some constants $c > 0$ such that, for any given $N > 4$, a strictly decreasing gap sequence $L = (l_1, l_2, \cdots, l_k)$ with $l_1 < N, l_k = 1$, there exists a complex number $z$ satisfying $z^N = 1$ and $|z - 1| > cN^{-1/2}$ such that for any permutation $\sigma$ of $1, 2, \cdots, N$ we have
    \begin{align}
        T(N, \sigma, L) \geq \frac{1}{4\pi} \cdot \frac{|p(N, \sigma, z) - p(N, id, z)|}{\min_{1 \leq s \leq k-1}(l_{s+1} + N^{1-\frac{1}{2s}})}
    \end{align}
\end{corollary}

\noindent The proof of this follows directly from theorem 4. Combining this with corollary 1 gives:

\begin{theorem}
    There exists some universal constant $C_1 > 0$ such that, for any given large $N$, a strictly decreasing gap sequence $L = (l_1, l_2, \cdots, l_k)$ with $l_1 < N, l_k = 1$, there exists a permutation $\sigma$ of $1, 2, \cdots, N$, such that
    \begin{align}
        T(N, \sigma, L) > \frac{C_1N^2}{\min_{1 \leq s \leq k-1}(l_{s+1} + N^{1-\frac{1}{2s}})}
    \end{align}
\end{theorem}

\noindent Finally, we give a requirement for a gap sequence to be able to achieve a worst-case time of $N$ times a polynomial of $\log N$:
 
\begin{theorem}
    Let $L_1, L_2, \cdots, L_n, \cdots$ be a sequence of strictly decreasing gap sequences where $L_i$ begins with at most $i$ for every $i \geq 1$. Denote the $j$th term of $L_i$ as $l_{i, j}$. If the worst case time of running shellsort on any array of length $N$ using $L_N$ is $O(N(\log N)^\mu)$ (where $\mu \geq 1$ is any fixed positive real number) as $N$ goes to infinity, then for any fixed $\epsilon > 0$ and letting $s_N =   \lfloor \log N/((2\mu+\epsilon)\log \log N) \rfloor$ we must have $l_{N, s_N} = \Omega(N/(\log N)^\mu)$ as $N$ goes to infinity, where the implied constant is independent of the choice of $\epsilon$ (i.e. only depending on the choice of the $L_n$'s).   
\end{theorem}

\begin{proof}
    Consider the denominator in (8) with $L = L_N$ given by $\min_{1 \leq s \leq k-1}(l_{N,  s+1} + N^{1-\frac{1}{2s}})$ where $k$ is the length of $L_N$. We want this number to be at least $DN/(\log N)^\mu$ for some constants $D,\mu > 0$. Fix $\epsilon > 0$, take $s = s_N - 1$. This gives $l_{N,s_N} + N/(\log N)^{\mu + \epsilon/2} \geq DN/(\log N)^\mu$, so $l_{N, S_N} \geq (D - o(1))N/(\log N)^\mu$, and we're done.
\end{proof} 

\section*{Acknowledgements}
\noindent The author would like to thank Professor Stefan Kiefer at St John's College, the University of Oxford for his insightful feedbacks and advices on an earlier draft of this paper.

\section*{Appendix: Full proof of lemma 1}

\begin{proof}
    Consider a quarter circle of the unit complex circle which contains the maximum amount of points among $z_1, \cdots, z_n$. Assume without loss of generality that this is the sector which can be bisected by the zero radian line and say that it contains $\lambda n$ of the points where $\lambda \geq 1/4$, so all of these points have real parts at least $\sqrt{2}/2$. So we consider all points on the unit circle whose corresponding complex number has a negative real part; this is the semicircle which is given by the inequality $\pi/2 < \arg(z) < 3\pi/2$, then we get that this sector contains at most $2\lambda n$ (or $(1-\lambda)n$ in case $\lambda > 1/3$) of the points. Also note that out of these points, at most $\lambda n$ of them lie in the quarter-circle $3\pi/4 < \arg (z) < 5\pi/4$, so the remaining points must have real points of at least $-\sqrt{2}/2$. We assume that $z_1, ... z_n$ are already ordered in terms of increasing real parts. If $1/4 \leq \lambda \leq 1/3$ we therefore have:
    \begin{align*}
        &\mathfrak{R}(z_1+2z_2+\cdots+nz_n)\\ &\geq -(1+2+\cdots+\lambda n) - (\sqrt{2}/2)((\lambda n+1)+\cdots+2\lambda n) + (\sqrt2/2)((1-\lambda)n+1 + \cdots + n) + O(n) \\
        & =-\frac{\lambda^2n^2}{2} - \frac{3\sqrt{2}\lambda^2n^2}{4} + \frac{\sqrt{2}(1-\lambda/2)\lambda n^2}{2} + O(n) \\ 
        & = (-\lambda^2/2 - \sqrt{2}\lambda^2 + \sqrt{2}\lambda/2)n^2 + O(n) > 0.01n^2
    \end{align*}

    \noindent and if $1/3 \leq \lambda \leq 1/2$, 
    \begin{align*}
        &\mathfrak{R}(z_1+2z_2+\cdots+nz_n)\\ &\geq -(1+2+\cdots+\lambda n) - (\sqrt{2}/2)((\lambda n+1)+\cdots+(1-\lambda) n) + (\sqrt2/2)((1-\lambda)n+1 + \cdots + n) + O(n) \\
        & =-\frac{\lambda^2n^2}{2} - \frac{\sqrt{2}(1-2\lambda) n^2}{4} + \frac{\sqrt{2}(1-\lambda/2)\lambda n^2}{2} + O(n) > 0.01n^2
    \end{align*}

    \noindent and finally if $\lambda > 1/2$:
    \begin{align*}
        \mathfrak{R}(z_1+2z_2 + \cdots + nz_n) &\geq -(1 + 2 + \cdots + n/2) + (\sqrt{2}/2)(n/2 + 1 + \cdots + n) \\ &= (3\sqrt{2}/16-1/8)n^2 + O(n) 
    \end{align*}
    and the lemma follows. 
\end{proof} 

\section*{Declaration of AI use}
OpenAI's ChatGPT has been used for providing the initial structural framework of the proof of Lemma 2, which has been subsequently verified, refined and written by the author.

\bibliographystyle{plain}

\end{document}